\documentclass{amsart}
\usepackage{fullpage,amsmath,amsthm,amsfonts,amssymb,mathrsfs,graphicx}
\newtheorem{proposition}{Proposition}
\newtheorem{theorem}{Theorem}
\def\Xint#1{\mathchoice
{\XXint\displaystyle\textstyle{#1}}%
{\XXint\textstyle\scriptstyle{#1}}%
{\XXint\scriptstyle\scriptscriptstyle{#1}}%
{\XXint\scriptscriptstyle\scriptscriptstyle{#1}}%
\!\int}
\def\XXint#1#2#3{{\setbox0=\hbox{$#1{#2#3}{\int}$ }
\vcenter{\hbox{$#2#3$ }}\kern-.57\wd0}}

\def\dashint{\Xint-}

\title{The Benjamin-Ono Hierarchy with Asymptotically Reflectionless Initial Data in the Zero-Dispersion Limit}
\author[P. D. Miller]{Peter D. Miller}
\address{Department of Mathematics\\University of Michigan\\East Hall\\
530 Church St.\\
Ann Arbor, MI 48109} 
\email[P. D. Miller]{millerpd@umich.edu}
\urladdr[P. D. Miller]{http://www.math.lsa.umich.edu/~millerpd}
\thanks{This work was supported by the National Science Foundation under grant DMS-0807653.}
\author[Z. Xu]{Zhengjie Xu}
\email[Z. Xu]{zhengjxu@umich.edu}
\curraddr[Z. Xu]{Bloomberg L. P.}
\date{\today}
\begin{document}
\maketitle
\begin{abstract}
We study the Benjamin-Ono hierarchy with positive initial data of a general type, in the limit when the
dispersion parameter tends to zero.  We establish simple formulae for the limits (in appropriate weak or distributional senses) of an infinite family of simultaneously conserved densities in terms of alternating sums of branches of solutions of the inviscid Burgers hierarchy.
\end{abstract}

\section{Introduction}
The Benjamin-Ono equation is a nonlinear evolution equation governing
certain types of 
\emph{internal waves}.  Internal waves are
disturbances --- set into motion by gravity --- of the
interface between two immiscible fluids of different densities. 
A number of assumptions are employed to deduce the Benjamin-Ono equation
as a simplified model 
from the full equations of three-dimensional fluid mechanics:  
\begin{itemize}
\item One supposes that 
both fluid layers consist of inviscid and incompressible fluids,
  and (for stability) the less-dense fluid rests on top of the denser
  fluid.
\item 
One supposes that the waves are propagating in one direction only, which
reduces the problem to that of two-dimensional fluid mechanics (only the
vertical direction and the propagation direction survive).
\item
One supposes that the top layer containing the less-dense fluid is thin 
compared to a typical wavelength of the interface.  This allows the 
top layer to be treated by a depth-averaging approach.  The top of this
layer is idealized to a rigid horizontal lid.
\item 
One supposes that the bottom layer containing the denser fluid is infinitely
thick.  This simplifies the contribution from this layer to 
the dispersion relation of small-amplitude waves on the interface in the
linear approximation.
\item
One supposes that the amplitude of the waves is small compared to the
thickness of the top layer.  This allows the nonlinear effects to be brought in
perturbatively, and also means we are assuming the interface does not
breach the surface by meeting the rigid lid.
\end{itemize}
In the linear and dispersionless approximation, deformations of the
interface satisfy the one-dimensional wave equation.  If one
perturbatively introduces the balanced effects of weak nonlinearity
and dispersion on a solution of the wave equation propagating to the
right (say) at constant velocity, then the Benjamin-Ono equation
arises as the first correction in the moving frame of reference of the
wave, describing the slow variation of the dimensionless wave height
$u$ as a function of a spatial coordinate $x$ in the propagation direction
and the time $t$:
\begin{equation}
u_t +2uu_x+\epsilon\mathcal{H}[u_{xx}]=0.
\label{eq:BO}
\end{equation}
Here, subscripts denote partial derivatives, and $\mathcal{H}$ is the
Hilbert transform operator defined on $L^2(\mathbb{R})$ by the singular integral
\begin{equation}
\mathcal{H}[f](x):=\frac{1}{\pi}\dashint_\mathbb{R}\frac{f(y)\,dy}{y-x}, 
\end{equation}
and $\epsilon>0$ is a dimensionless
measure of the relative strength of dispersive effects compared with
nonlinear effects.  If an initial condition $u=u_0(x)$ is given at $t=0$ where $u_0$ is independent
of $\epsilon$, then of course the corresponding solution of \eqref{eq:BO} will depend
on $\epsilon$, but we will often not be explicit about this dependence in our notation.  

It is quite useful in applications to have accurate and easily analyzable
models for internal waves.  Indeed, one application that is particularly
timely is the modeling of submerged ``plumes'' of oil as were reported
following the Deepwater Horizon leak in the Gulf of Mexico in May--August 2010.  Recent
experiments \cite{CamassaM10} performed by Roberto Camassa and Richard McLaughlin at the
University of North Carolina, Chapel Hill have demonstrated that if oil
is emitted as a turbulent jet from an ocean floor leak and a density
stratification is present in the surrounding fluid, then most of the oil
will become trapped at the interface between the dense and less-dense fluid
layers rather than floating to the surface, \emph{even though the oil
is less dense still than the upper layer}.  
In these circumstances, modeling the motion of the submerged oil plumes
within the fluid column amounts to modeling the motion of the density
interface, that is, modeling internal waves.  

The purpose of this short paper is to show how methods we have recently developed \cite{MillerX10} to study the
asymptotic behavior of solutions of the Benjamin-Ono equation as $\epsilon\downarrow 0$ in
the weak topology extend both to the whole hierarchy of  ``higher-order'' Benjamin-Ono equations and also to the whole hierarchy of densities coming from conservation laws.  It is a pleasure to be able
to contribute to this special volume of papers in honor of Dave Levermore.  It will be clear to all of those
who have followed his work that our small contribution is directly inspired by his groundbreaking analysis with Peter Lax of the zero-dispersion limit for the Korteweg-de Vries equation \cite{LaxL83}, a project that
has had a tremendous impact on fields of study ranging from asymptotic analysis of nonlinear evolution equations to the theory of orthogonal polynomials and of random matrices.  

\section{The Benjamin-Ono Hierarchy}
Let us take the phase space of fields $u$ to be
\begin{equation}
\mathscr{P}:=\{u\in L^2(\mathbb{R})\cap C^\infty(\mathbb{R}),\; u^{(k)}\in L^2(\mathbb{R}),\;\forall k\}.
\end{equation}
This is clearly a linear space over $\mathbb{R}$, and it is also an algebra that is closed under differentiation and Hilbert transforms.  We will have use below for the Cauchy operators $\mathcal{C}_\pm$ densely defined on $L^2(\mathbb{R})$ by singular integrals as follows:
\begin{equation}
\mathcal{C}_\pm[f](x):=\lim_{\delta\downarrow 0}\frac{1}{2\pi i}\int_\mathbb{R}\frac{f(y)\,dy}{y-x\mp i\delta}.
\end{equation}
The Cauchy operators are bounded with respect to the $L^2(\mathbb{R})$ norm and hence extend to bounded operators with the same norms on all of $L^2(\mathbb{R})$.  In fact, the operators $\mathcal{C}_+$ and $-\mathcal{C}_-$ are just the self-adjoint orthogonal projections from $L^2(\mathbb{R})$ onto
the Hardy subspaces of functions analytic in the upper and lower half-planes, respectively.  They satisfy the identities
\[
\mathcal{C}_+\circ(-\mathcal{C}_-)=(-\mathcal{C}_-)\circ\mathcal{C}_+=0,\quad
\mathcal{C}_+^2=\mathcal{C}_+,\quad (-\mathcal{C}_-)^2=-\mathcal{C}_-,\quad
\mathcal{C}_+-\mathcal{C}_-=1,\quad \mathcal{C}_++\mathcal{C}_-=-i\mathcal{H}.
\]
Since $\mathcal{C}_\pm$ are self-adjoint, it follows that $\mathcal{H}$ is skew-adjoint, that is,
\begin{equation}
\int_\mathbb{R}f\mathcal{H}[g]\,dx = -\int_\mathbb{R}g\mathcal{H}[f]\,dx
\end{equation}
whenever $f$ and $g$ area real-valued functions in the phase space $\mathscr{P}$.
The operators $\mathcal{C}_\pm$ and $\mathcal{H}$ commute with differentiation in $x$.
Also, for functionals $I[u]$ defined on the phase space $\mathscr{P}$, we define the variational
derivative $\delta I/\delta u$ by 
\begin{equation}
\left.\frac{d}{dt}I[u+ tv]\right|_{t=0}=\int_\mathbb{R}\frac{\delta I}{\delta u}[u](x)v(x)\,dx.
\end{equation}
For functionals $I$ with $I[0]=0$, assuming the existence of the variational derivative of $I$ for each $u\in\mathscr{P}$ we can recover
the functional from its derivative by the formula
\begin{equation}
I[u] = \int_0^1\int_\mathbb{R}\frac{\delta I}{\delta u}[tu](x)u(x)\,dx\,dt = 
\int_\mathbb{R}\left(\int_0^1\frac{\delta I}{\delta u}[tu](x)\,dt\right) u(x)\,dx.
\label{eq:integratebackup}
\end{equation}

\subsection{Conservation Laws for the Benjamin-Ono Equation}
The Benjamin-Ono equation conserves an infinite number of functionals of $u$.  These may be obtained by several different methods, several of which we review (in historical order of discovery) for the reader's convenience.
\subsubsection{The Nakamura Scheme}
A.~Nakamura \cite{Nakamura79} (see also \cite{Matsuno}) was the first to deduce an infinite number of conserved quantities for the Benjamin-Ono equation \eqref{eq:BO}.  His derivation is  based on a 
B\"acklund
transformation for \eqref{eq:BO}.  The B\"acklund transformation is
\begin{equation}
-i\epsilon\mu\mathcal{C}_+[n_x] + 1-e^{-n} = \mu u
\label{eq:Baecklund}
\end{equation}
where $\mu$ is an arbitrary parameter.  From this equation it can be shown that regardless of the value of $\mu$, 
\begin{equation}
\frac{d}{dt}\int_\mathbb{R}n\,dx = 0
\end{equation}
when $u$ satisfies \eqref{eq:BO}.
Therefore, by expanding $n$ in a power series in $\mu$ with coefficients depending on $u$, the coefficients
will all be densities of conserved functionals of $u$.  Writing $n=\mu n_1+\mu^2n_2+\mu^3n_3+\cdots$ one easily obtains a recurrence in which the conserved density $n_n$ is explicitly given in terms of 
$n_1,n_2,\dots,n_{n-1}$ and $u$.  In fact, the recurrence can be made more explicit by noting that the nonlinearity of the scheme is only quadratic in $n_x$:  differentiating \eqref{eq:Baecklund} with respect to $x$ and then using \eqref{eq:Baecklund} to eliminate $e^{-n}$ from the result one arrives at
\begin{equation}
n_x =\mu\left(un_x + u_x+i\epsilon n_x\mathcal{C}_+[n_x]+i\epsilon\mathcal{C}_+[n_{xx}]\right),
\end{equation}
and therefore $n_{1,x}=u_x$ and 
\begin{equation}
n_{m,x}=un_{m-1,x} + i\epsilon\mathcal{C}_+[n_{m-1,xx}] + i\epsilon\sum_{j=1}^{m-2} n_{j,x}\mathcal{C}_+[n_{m-1-j,x}],\quad m=2,3,4,\dots.
\end{equation}
The first several densities are given by:
\begin{equation}
\begin{split}
n_1 &=u\\
n_2 &= \frac{1}{2}u^2 +\epsilon\left( \frac{1}{2}\mathcal{H}[u_x] +\frac{1}{2} iu_x\right)\\
n_3&= \frac{1}{3}u^3 +\epsilon\left(\frac{1}{2} u\mathcal{H}[u_x] + i uu_x +\frac{1}{2}\mathcal{H}[uu_x]
\right) +\epsilon^2\left(\frac{1}{2}i\mathcal{H}[u_{xx}]-\frac{1}{2} u_{xx}
\right) \\
n_4 &=\frac{1}{4}u^4 +\epsilon\left(\frac{1}{2}u^2\mathcal{H}[u_x]+\frac{1}{2}u\mathcal{H}[uu_x]
\right)-\epsilon^2\frac{1}{2}uu_{xx} \\
&\qquad{}+\left[\epsilon\left(\frac{1}{2}iu^3+\frac{1}{6}\mathcal{H}[u^3]\right)+\epsilon^2
\left(\frac{3}{4}iu\mathcal{H}[u_x]-\frac{3}{4}uu_x+\frac{1}{4}\mathcal{H}[u\mathcal{H}[u_x]]\right)\right]_x\\
&\quad\qquad{}
-\epsilon^2\frac{1}{4}u_x\mathcal{H}[u_x]+\epsilon^2\frac{1}{8}\left(\mathcal{H}[u_x]^2-(u_x)^2\right])
\\
&\qquad\qquad{}+\epsilon^2
\frac{3}{4}i\mathcal{H}\left[(u_x)^2+uu_{xx}\right].
\end{split}
\end{equation}
In the expression for $n_4$, only the terms on the first line contribute to the integral over $\mathbb{R}$.
Indeed, those on the second line are derivatives of functions in $\mathscr{P}$, and those on the third line
have zero integral because $\mathcal{H}$ is skew-adjoint and $\mathcal{H}^2=-1$.  To deduce that
the terms on the fourth line have zero integral, it is necessary to correctly interpret the integral as the integrand is not of class $L^1(\mathbb{R})$.  The key identity here is the following:  if $f\in\mathscr{P}$
then
\begin{equation}
\lim_{R\uparrow\infty}\int_{-R}^R\mathcal{H}[f](x)\,dx = 0.
\label{eq:intHzero}
\end{equation}

\subsubsection{The Fokas-Fuchssteiner Scheme}
The following method is due to Fokas and Fuchssteiner \cite{FokasF81}.  It is based on Lie-theoretic analysis of one-parameter symmetry groups of the Benjamin-Ono equation \eqref{eq:BO}.  
One begins with 
\begin{equation}
\frac{\delta I_1}{\delta u}[u] = 1\quad\text{and}\quad\frac{\delta I_2}{\delta u}[u]=u
\end{equation}
and then recursively defines
\begin{equation}
\frac{\delta I_m}{\delta u}[u]=\frac{1}{m-2}\frac{\delta}{\delta u}\int_{\mathbb{R}}
\left[2xuu_x+u^2 +
\epsilon\left(x\mathcal{H}\left[u_{xx}\right]+\frac{3}{2}
\mathcal{H}\left[u_x\right]\right)\right]\frac{\delta I_{m-1}}{\delta u}[u]\,dx,\quad m=3,4,5,\dots.
\label{eq:FFrecursion}
\end{equation}
To compute the variational derivative on the right-hand side one needs the identity
\begin{equation}
\mathcal{H}[xf]=x\mathcal{H}[f]\quad\text{if}\quad \int_\mathbb{R}f(x)\,dx=0.
\end{equation}
Although the explicit function $x$ appears in the integrand on the right-hand side of \eqref{eq:FFrecursion}, the recursion guarantees that the
variational derivatives produced are all generated from $u$ and its derivatives and a finite number of
applications of $\mathcal{H}$.  
For example, 
\begin{equation}
\frac{\delta I_3}{\delta u}[u]=\frac{\delta}{\delta u}\int_\mathbb{R}\left[2xu^2u_x+u^3 +
\epsilon\left(xu\mathcal{H}[u_{xx}]+\frac{3}{2}u\mathcal{H}[u_x]\right)\right]\,dx=u^2+\epsilon\mathcal{H}[u_x]
\end{equation}
and
\begin{equation}
\frac{\delta I_4}{\delta u}[u] = u^3+\epsilon\left(\frac{3}{2}u\mathcal{H}[u_x]+\frac{3}{2}\mathcal{H}[uu_x]\right)-\epsilon^2u_{xx}.
\end{equation}
As $x$ does not appear explicitly in the resulting expressions, it follows from the formula
\eqref{eq:integratebackup} that the functionals $I_m$ are integrals of densities that also do not
involve $x$ explicitly.  These densities may be easily obtained from the variational derivatives simply
by first multiplying through term-by-term by $u$ and then dividing each term by its
homogeneous degree in $u$.  For example, from each $\delta I_m/\delta u$ we obtain a corresponding
conserved density $f_m$ as follows:
\begin{equation}
\begin{split}
f_1&=u\\
f_2 &=\frac{1}{2}u^2\\
f_3 &=\frac{1}{3}u^3 + \epsilon\frac{1}{2}u\mathcal{H}[u_x]\\
f_4&=\frac{1}{4}u^4 + \epsilon\left(\frac{1}{2}u^2\mathcal{H}[u_x]+\frac{1}{2}u\mathcal{H}[uu_x]\right)
-\epsilon^2\frac{1}{2}uu_{xx}.
\end{split}
\end{equation}
Unlike the densities produced by the Nakamura scheme, these densities are all absolutely integrable if $u\in\mathscr{P}$.

\subsubsection{The Kaup-Matsuno Scheme}
The following method was derived from the inverse-scattering transform for the Benjamin-Ono equation by  Kaup and Matsuno \cite{KaupM98}.  Set $k_1:=u$ and then define recursively
\begin{equation}
k_{m}:=u\mathcal{C}_+[k_{m-1}]+i\epsilon\left(\frac{k_{m-1}}{u}\right)_x,\quad m=2,3,\dots.
\end{equation}
Then, the quantities $k_m$ are all densities of functionals conserved by \eqref{eq:BO}. 
The first several densities obtained by the Kaup-Matsuno scheme are
\begin{equation}
\begin{split}
k_1 &= u\\
k_2 &= \frac{1}{2}u^2-\frac{1}{2}i u\mathcal{H}[u]\\
k_3 &=\frac{1}{4}u^3-\frac{1}{4}u\mathcal{H}[u\mathcal{H}[u]]-\frac{1}{4}iu^2\mathcal{H}[u]-
\frac{1}{4}iu\mathcal{H}[u^2]+\epsilon\left(\frac{1}{2}iuu_x+\frac{1}{2}u\mathcal{H}[u_x]\right)\\
k_4 &=\frac{1}{8}u^4-\frac{1}{8}u^2\mathcal{H}[u\mathcal{H}[u]]-\frac{1}{8}u\mathcal{H}[u^2\mathcal{H}[u]]-\frac{1}{8}u\mathcal{H}[u\mathcal{H}[u^2]]+\frac{1}{8}u\mathcal{H}[u\mathcal{H}[u\mathcal{H}[u]]]\\
&\qquad{}-\frac{1}{8}iu^3\mathcal{H}[u]-\frac{1}{8}iu^2\mathcal{H}[u^2]-\frac{1}{8}iu\mathcal{H}[u^3]\\
&\quad\qquad{}+\epsilon\left(\frac{1}{2}u^2\mathcal{H}[u_x]+\frac{3}{4}u\mathcal{H}[uu_x]+\frac{1}{4}uu_x\mathcal{H}[u]+\frac{3}{4}iuu_x-\frac{1}{2}iu\mathcal{H}[u\mathcal{H}[u_x]]-\frac{1}{4}iu\mathcal{H}[u_x\mathcal{H}[u]]\right)\\
&\qquad\qquad{}+\epsilon^2\left(-\frac{1}{2}uu_{xx}+\frac{1}{2}iu\mathcal{H}[u_{xx}]\right).
\end{split}
\end{equation}
As is the case for the Fokas-Fuchssteiner scheme, and unlike the Nakamura scheme, the densities
produced by the Kaup-Matsuno scheme are all absolutely integrable for $u\in \mathscr{P}$.
\subsubsection{Comparison of the Schemes}
Although the three schemes clearly do not give rise to identical densities, they apparently produce exactly the same integrals.  That is, for $u\in\mathscr{P}$ we typically have $n_m$, $f_m$, and $k_m$
being different (unequal) expressions generated by derivatives of $u$ and applications of $\mathcal{H}$.  However, 
\begin{equation}
I_m[u]:=\int_\mathbb{R}n_m\,dx = \int_{\mathbb{R}}f_m\,dx = \int_{\mathbb{R}}k_m\,dx,\quad m=1,2,3,\dots.
\end{equation}
Here in the case of $\int_\mathbb{R}n_m\,dx$ the integral must generally be interpreted in the ``principal value at infinity'' sense as in \eqref{eq:intHzero}.
From each of the three schemes it follows that if $u$ is a smooth solution of \eqref{eq:BO}, then $dI_m/dt=0$ for all $m$.  
We do not provide a direct proof of the equivalence of the integrals generated by each of the three schemes here, although 
in the case where $\epsilon=0$ all three schemes produce the same result:
if $u\in\mathscr{P}$ is  independent of $\epsilon$, then
\begin{equation}
\lim_{\epsilon\downarrow 0}I_m(t) = \frac{1}{m}\int_\mathbb{R}u(x)^m\,dx.
\end{equation}
This is rather straightforward to show from the Nakamura and Fokas-Fuchssteiner schemes, while for the Kaup-Matsuno scheme it follows from a lemma proved in the appendix of \cite{MillerX10}.
In any case, this fact easily establishes the functional independence of the integrals $I_m[u]$, $m=1,2,3,\dots$.

The variational derivatives $\delta I_m/\delta u$ in fact play a dual role, as is proved in
both \cite{FokasF81} and \cite{Matsuno} (in the respective context of two different schemes):  they also serve as densities of integrals:
\begin{equation}
\int_{\mathbb{R}}\frac{\delta I_m}{\delta u}\,dx = (m-1)I_{m-1}[u],\quad m\ge 2, \quad
\text{or}\quad
I_m[u] = \int_\mathbb{R}\frac{1}{m}\frac{\delta I_{m+1}}{\delta u}\,dx,\quad m\ge 1.
\label{eq:varderivdensities}
\end{equation}

The main advantage of the Nakamura scheme is that it can be used to place all of the equations of
the Benjamin-Ono hierarchy (see below) in bilinear form after which Hirota's method can be applied
to deduce the form of the simultaneous $N$-soliton solution of the entire hierarchy.  This is 
done in \cite{Matsuno}, and we will use the resulting formulae below.  However, due to the presence
of terms of the form $\mathcal{H}[f]$ for $f\in\mathscr{P}$ that are not themselves of the form $f=g_x$ for $g\in\mathscr{P}$, the densities generated by the Nakamura scheme are generally not absolutely integrable and require a more careful interpretation of the integral.

The advantages of the Fokas-Fuchssteiner scheme
are (i) that it provides a direct one-term recurrence for the variational derivatives of the 
conserved quantities (which unlike densities are uniquely determined by the functional),
(ii) that the variational derivatives it generates are manifestly real and the corresponding
densities $f_m$ are simpler than in either of the other two schemes, 
and (iii) that it allows a direct proof of the fact that the functionals $I_m$ are all in involution with respect to the Poisson bracket:
\begin{equation}
\{I,J\}:=\int_\mathbb{R}\frac{\delta I}{\delta u}(x)\frac{\partial}{\partial x}\frac{\delta J}{\delta u}(x)\,dx,
\end{equation}
in other words, $\{I_j,I_k\}=0$ for all $j,k$.  

In the Kaup-Matsuno scheme, the quantities $\overline{N}_m:=k_m/u$ are the coefficients in the Laurent expansion
 about $\lambda=\infty$ of the eigenfunction $\overline{N}(\lambda;x,t)$ that is a simultaneous solution of the two
linear equations making up the Lax pair for the Benjamin-Ono equation.  Here $\lambda\in\mathbb{C}$ is the spectral parameter.  Because the integrals $I_m[u]:=\int_\mathbb{R}k_n\,dx$ are obtained by expanding a scattering eigenfunction, they equivalently encode the scattering data in an explicit way, and in \cite{KaupM98} one can find explicit formulae for $I_m$ in terms of the discrete spectrum and the reflection coefficient
of the direct scattering problem.  Thus, the Kaup-Matsuno scheme leads not just to an infinite collection of integrals of motion of the Benjamin-Ono equation, but also a hierarchy of trace formulae equivalently expressing each $I_m$ both as a functional of $u$ and also as functional of the 
scattering data.  These formulae can be used to deduce from initial data the asymptotic distribution of eigenvalues $\lambda_j$
of the scattering problem in the zero-dispersion limit (see \cite{Matsuno} and section 3.2 of \cite{MillerX10}).

\subsection{Construction of the Hierarchy}
The Benjamin-Ono equation \eqref{eq:BO} can be written in Hamiltonian form as
\begin{equation}
u_t=-\frac{\partial}{\partial x}\frac{\delta I_3}{\delta u}.
\end{equation}
The Noetherian symmetry of \eqref{eq:BO} associated with the conserved quantity $I_m$ is the Hamiltonian flow with Hamiltonian $I_m$:
\begin{equation}
u_{t_k} = -\frac{\partial}{\partial x}\frac{\delta I_{k+2}}{\delta u},\quad k=1,2,3,\dots.
\label{eq:Hierarchy}
\end{equation}
Here $t_k$ is the parameter of the symmetry group generated by $I_{k+2}$.  The fact that the
integrals $I_k$ are all in involution \cite{FokasF81} implies that these flows are all compatible (that is, the symmetry group is abelian), so given a smooth function $u_0\in \mathscr{P}$ and a positive integer $K$, 
there will exist a function $u(x,t_1,t_2,\dots,t_K)$ satisfying $u(x,0,0,\dots,0)=u_0(x)$
and equations \eqref{eq:Hierarchy} for $k=1,2,\dots,K$.  Equations \eqref{eq:Hierarchy} constitute
the Benjamin-Ono hierarchy.  

\section{Setting up the Zero-Dispersion Limit}
\subsection{Formulation of the Problem}
The problem we wish to consider is the following.  Let $u_0\in \mathscr{P}$ be given, an initial condition 
independent of $\epsilon$.  For each $\epsilon>0$ we may construct the simultaneous solution 
$u(x,t_1,t_2,\dots,t_K)$ of the Benjamin-Ono hierarchy \eqref{eq:Hierarchy} of commuting flows
satisfying the initial condition $u(x,0,0,\dots,0)=u_0(x)$.  The question of interest is the asymptotic behavior of $u(x,t_1,t_2,\dots,t_K)$ in the zero-dispersion limit $\epsilon\downarrow 0$.  
As a first step, we will address this problem by establishing the existence of the dispersionless limits (in appropriately weak topologies) of all of the 
conserved densities (see \eqref{eq:varderivdensities})
\begin{equation}
D_m:=\frac{1}{m}\frac{\delta I_{m+1}}{\delta u},\quad m=1,2,3,\dots.
\end{equation}
Note that $D_1=u$.  In general, $D_m$ differs from all three of $n_m$, $f_m$, and $k_m$ by a
``trivial'' density that integrates to zero for all $u\in\mathscr{P}$.  However, the densities $D_m$ are those that
most directly yield a dispersionless representation.  
\subsection{Admissible Initial Conditions}
We will further assume that $u_0\in\mathscr{P}$ satisfies the following conditions adapted from \cite{MillerX10}:
\begin{itemize}
\item $u_0(x)>0$ for all $x\in\mathbb{R}$.
\item There is a unique critical point $x_0\in\mathbb{R}$ for which $u_0'(x_0)=0$, and $u_0''(x_0)<0$
making $x_0$ the global, nondegenerate maximizer of $u_0$. 
\item $u_0$ exhibits power-law decay in its tails:  $\lim_{x\to\pm\infty}u_0(x)=0$ and 
\begin{equation}
\lim_{x\to\pm\infty}|x|^{q+1}u_0'(x)=C_\pm\quad\text{for some $q>1$,}
\end{equation}
where $C_+<0$ and $C_->0$ are constants.  
\item For each $k=1,2,\dots,K$ let $f(x)=u_0(x)^k$.  Then in each bounded interval there exist at most finitely many points $x=\xi$
at which $f''(\xi)=0$, and each is a simple inflection point:  $f'''(\xi)\neq 0$.
\end{itemize}
Such $u_0\in\mathscr{P}$ will be called \emph{admissible} initial conditions.  An example is shown in Figure~\ref{fig:IC}.
\begin{figure}[h]
\begin{center}
\includegraphics[width=0.3\linewidth]{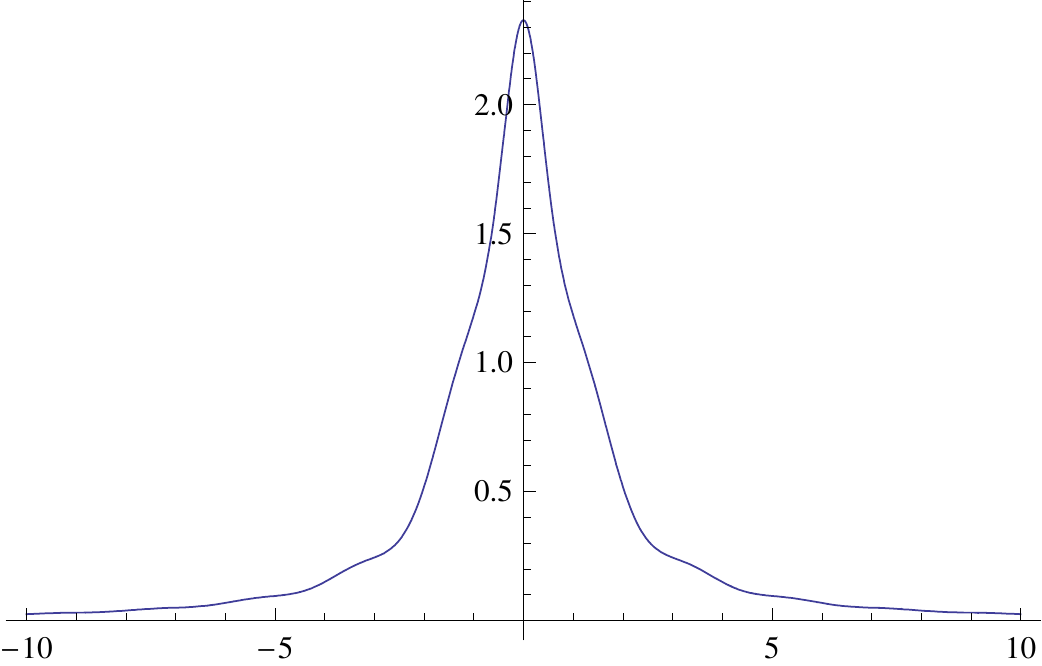}\hspace{0.049\linewidth}%
\includegraphics[width=0.3\linewidth]{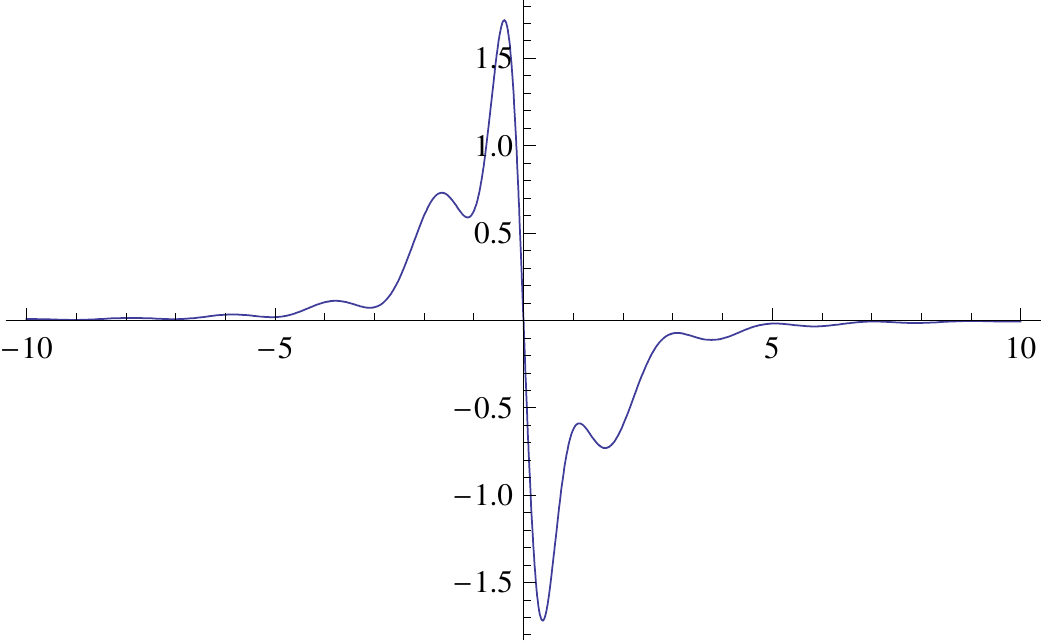}\hspace{0.049\linewidth}%
\includegraphics[width=0.3\linewidth]{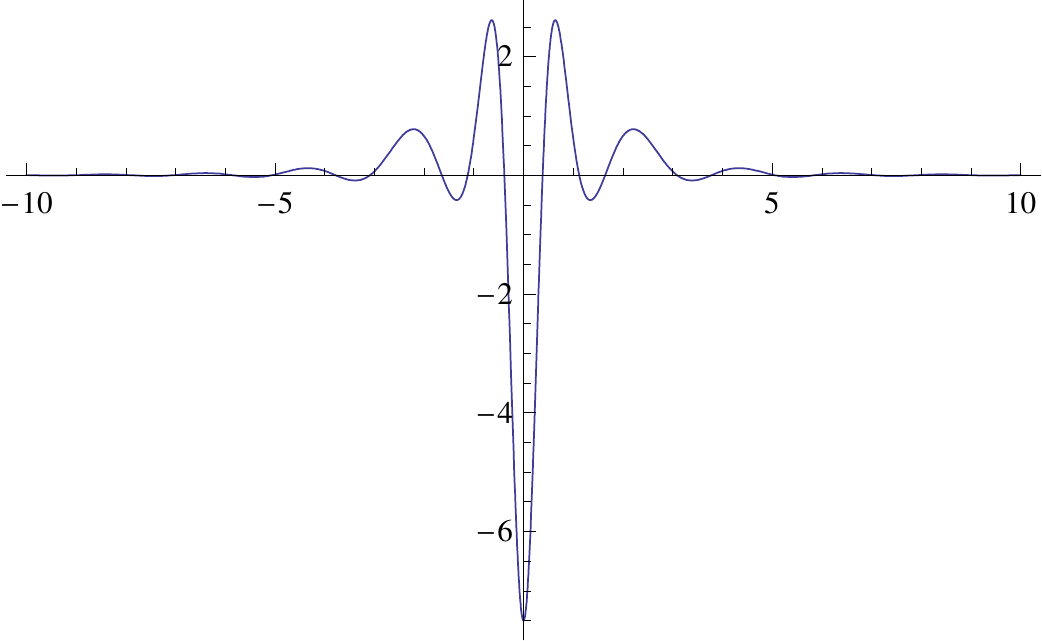}
\end{center}
\caption{An admissible initial condition for which $u_0'(x)=-5x(1+\cos(\pi x)/2)/(1+x^2)^2$.  Left:  $u_0(x)$.
Center:  $u_0'(x)$.  Right:  $u_0''(x)$.  Note that this particular admissible initial condition has an
infinite number of inflection points asymptotically near integer values of $x$ for $x$ large.}
\label{fig:IC}
\end{figure}

\subsection{Reflectionless Modification of the Initial Data}
Our method will be to study the Benjamin-Ono hierarchy with admissible initial condition $u_0$ using the inverse-scattering transform for the $x$-part of the Lax pair (see \cite{FokasA83,KaupM98}).  
The first step in the process is to associate to $u_0$ its scattering data, consisting of a complex-valued
reflection coefficient $\beta(\lambda)$, $\lambda>0$, as well as discrete eigenvalues $\{\lambda_j<0\}$ and corresponding phase constants $\{\gamma_j\in\mathbb{R}\}$.  Matsuno was the first to observe that the conservation laws can be used to deduce information 
about the scattering data (see \cite{Matsuno}, section 3.3) when the parameter $\epsilon>0$ (which appears parametrically in the scattering problem although the potential $u_0$ is independent of $\epsilon$) is small.  His analysis becomes more rigorous with the use of the trace formulae arising 
from the Kaup-Matsuno scheme.  For admissible initial conditions, Matsuno's main results are:
\begin{itemize}
\item The reflection coefficient $\beta(\lambda)$ is small when $\epsilon\ll 1$.  
\item The number $N$ of eigenvalues is large when $\epsilon\ll 1$, but 
\begin{equation}
\lim_{\epsilon\downarrow 0}\epsilon N = M:=\int_\mathbb{R}u_0(x)\,dx.
\label{eq:massdef}
\end{equation}
The number $N[a,b]$ of eigenvalues in the interval $[a,b]$, $-L\le a\le b\le 0$, $L:=\max_{x\in\mathbb{R}}u_0(x)$, satisfies
\begin{equation}
\lim_{\epsilon\downarrow 0}\epsilon N[a,b] = \int_a^b F(\lambda)\,d\lambda,\quad F(\lambda):=\frac{1}{2\pi}(x_+(\lambda)-x_-(\lambda)).
\end{equation}
Here $x_-(\lambda)<x_+(\lambda)$ are defined for $-L<\lambda<0$ as the two solutions of the
equation $u_0(x)=-\lambda$.  They play the role of turning points in this theory.
\end{itemize}
(Note that the ``mass'' $M$ defined by \eqref{eq:massdef} is finite for admissible $u_0$ although it is
not so for general elements of $\mathscr{P}$.)  The solution of the Benjamin-Ono hierarchy for an admissible initial condition is therefore
in particular approximately reflectionless in the zero-dispersion limit.  In the absence of reflection 
the exact solution of the hierarchy is a multi-soliton solution that takes the form \cite{Matsuno}:
\begin{equation}
u(x,t_1,t_2,\dots,t_K)=2\epsilon\frac{\partial}{\partial x}\Im\{\log(\tau(x,t_1,t_2,\dots,t_k))\}
\label{eq:ureflectionless}
\end{equation}
where the ``tau function'' is
\begin{equation}
\tau(x,t_1,t_2,\dots,t_K):=\det(\mathbb{I}+i\epsilon^{-1}\mathbf{A})
\end{equation}
and $\mathbf{A}$ is an $N\times N$ Hermitian matrix with constant off-diagonal elements
\begin{equation}
A_{nm}:=\frac{2i\epsilon\sqrt{\lambda_n\lambda_m}}{\lambda_n-\lambda_m},\quad n\neq m
\label{eq:Aoffdiag}
\end{equation}
and diagonal elements depending explicitly on $x, t_1,t_2,\dots,t_K$:
\begin{equation}
A_{nn}:=-2\lambda_n\left(X(\lambda_n;x,t_1,t_2,\dots,t_K)+\gamma_n\right),
\label{eq:Adiag}
\end{equation}
where for future convenience we define a polynomial in $\lambda$ by
\begin{equation}
X(\lambda;x,t_1,t_2,\dots,t_K):=x-\sum_{k=1}^K(k+1)(-\lambda)^kt_k.
\end{equation}

To specify an appropriate family of exact solutions of the Benjamin-Ono hierarchy, first define the
exact number of approximate eigenvalues by 
\begin{equation}
N(\epsilon):=\left\lfloor\frac{M}{\epsilon}\right\rfloor.
\end{equation}
Then, define approximations $\{\tilde{\lambda}_n\}_{n=1}^{N(\epsilon)}$ with $-L<\tilde{\lambda}_1<\tilde{\lambda}_2<\cdots <\tilde{\lambda}_{N(\epsilon)}<0$ by quantizing the Matsuno eigenvalue
density:
\begin{equation}
\int_{-L}^{\tilde{\lambda}_n}F(\lambda)\,d\lambda = \epsilon\left(n-\frac{1}{2}\right),\quad n=1,2,\dots,N(\epsilon).
\end{equation}
Next, define approximations for the phase constants $\{\gamma_n\}$ by setting \cite{MillerX10}:
\begin{equation}
\tilde{\gamma}_n:=\gamma(\tilde{\lambda}_n),\quad\gamma(\lambda):=-\frac{1}{2}(x_+(\lambda)+x_-(\lambda)),\quad -L\le \lambda<0.
\end{equation}
Now, for each $\epsilon>0$, let $\tilde{u}=\tilde{u}(x,t_1,t_2,\dots,t_K)$ denote the exact
solution of the Benjamin-Ono hierarchy given by the reflectionless solution formula \eqref{eq:ureflectionless} with determinantal tau function $\tilde{\tau}$ involving the $N(\epsilon)\times N(\epsilon)$ Hermitian matrix $\tilde{\mathbf{A}}$ whose elements are given by \eqref{eq:Aoffdiag}--\eqref{eq:Adiag} with $\lambda_n$ and $\gamma_n$ replaced by $\tilde{\lambda}_n$ and $\tilde{\gamma}_n$ respectively, for $1\le n\le N(\epsilon)$.   In \cite{MillerX10} it is proved that for
admissible $u_0$, 
\begin{equation}
\lim_{\epsilon\downarrow 0}\int_\mathbb{R}|u_0(x)-\tilde{u}(x,0,0,\dots,0)|^2\,dx = 0,
\end{equation}
so that the replacement of the scattering data we have just made amounts to a modification of the initial condition that is negligible in the $L^2(\mathbb{R})$ sense in the zero-dispersion limit.

\section{Distributional Limits of Conserved Densities}
The conserved densities $\tilde{D}_m$, $m\ge 1$ (we are using tildes to remind the reader that these are expressions in $\tilde{u}$, its derivatives in $x$, and Hilbert transforms thereof), all have representations in terms of the tau function
associated with $\tilde{u}$.  Indeed, since $\tilde{u}$ satisfies \eqref{eq:Hierarchy} it follows from
using the formula \eqref{eq:ureflectionless} for $\tilde{u}$ that 
\begin{equation}
2\epsilon \frac{\partial^2}{\partial x\partial t_k}\Im\{\log(\tilde{\tau})\} = -(k+1)\frac{\partial}{\partial x}\tilde{D}_{k+1}.
\end{equation}
Integration in $x$ using decay at $x=\pm\infty$ to fix the integration constant yields the formulae
\begin{equation}
\tilde{D}_m = -\frac{2\epsilon}{m}\frac{\partial}{\partial t_{m-1}}\Im\{\log(\tilde{\tau})\},\quad m=2,3,\dots,K+1.
\end{equation}
Of course since $\tilde{D}_1=\tilde{u}$ a slightly different formula holds for $\tilde{D}_1$ according
to \eqref{eq:ureflectionless}.  In principle, this gives a way of evaluating $\tilde{D}_m$ for arbitrary $m$, although one must include dependence on a sufficient number of times $t_k$ by choosing $K$ large enough.

Let $\alpha_1\le\alpha_2\le\cdots\le\alpha_{N(\epsilon)}$ denote the (real) eigenvalues of
$\tilde{\mathbf{A}}$.  Then, we may write the densities in the form
\begin{equation}
\tilde{u}=\tilde{D}_1=\frac{\partial \tilde{U}}{\partial x},\quad \tilde{D}_m=-\frac{1}{m}\frac{\partial \tilde{U}}{\partial t_{m-1}},\quad m=2,3,\dots,K+1,
\end{equation}
where
\begin{equation}
\tilde{U}(x,t_1,t_2,\dots,t_K):=\epsilon\sum_{n=1}^{N(\epsilon)} 2\arctan(\epsilon^{-1}\alpha_n).
\label{eq:Uformula}
\end{equation}
Remarkably, the function $\tilde{U}$ has a completely explicit zero-dispersion limit:
\begin{proposition}
Uniformly on compact subsets of $\mathbb{R}^{K+1}$, 
\begin{equation}
\lim_{\epsilon\downarrow 0}\tilde{U}(x,t_1,t_2,\dots,t_K) =V(x,t_1,t_2,\dots,t_K):= \int_\mathbb{R}\pi\,\mathrm{sgn}(\alpha)
G(\alpha;x,t_1,t_2,\dots,t_K)\,d\alpha
\label{eq:Vdef}
\end{equation}
where
\begin{equation}
G(\alpha;x,t_1,t_2,\dots,t_K):=-\frac{1}{4\pi}\int_{-L}^0\chi_I(\alpha)\frac{d\lambda}{\lambda}
\label{eq:Gdef}
\end{equation}
and where $\chi_I(\alpha)$ denotes the indicator function of the interval
\begin{equation}
I:=\left[-2\lambda\left(X(\lambda;x,t_1,t_2,\dots,t_K)-x_+(\lambda)\right),
-2\lambda\left(X(\lambda;x,t_1,t_2,\dots,t_K)-x_-(\lambda)\right)\right].
\end{equation}
\end{proposition}
\begin{proof}
This is a simple generalization of Proposition~4.2 from \cite{MillerX10} and it is proved in exactly the same way (see in particular sections 4.1--4.3 of that reference).  For the reader's convenience we will
simply describe the idea of the proof.  

The key observation is that the eigenvalues of the matrix $\tilde{\mathbf{A}}$  have a limiting density $G(\alpha;x,t_1,t_2,\dots,t_K)$; that is, the normalized (to total mass $M$) counting measures
of eigenvalues of $\tilde{\mathbf{A}}$ converge in the weak-$\ast$ sense to $G\,d\alpha$ as $\epsilon\downarrow 0$.  This fact is 
proved using Wigner's method of moments.  One studies the asymptotic behavior of traces of arbitrary powers
of the $N(\epsilon)\times N(\epsilon)$ matrix $\tilde{\mathbf{A}}$ in the limit $N(\epsilon)\to\infty$,
and with the use of some combinatorial arguments and approximation in terms of diagonal and Toeplitz matrices one obtains leading-order asymptotic formulae for these, which in turn are proportional to moments of the eigenvalue counting measures.  Then one solves the moment problem for the limiting
moments to obtain $G\,d\alpha$.  Finally, by estimating the extreme eigenvalues of $\tilde{\mathbf{A}}$
one is able to convert convergence of moments to weak-$\ast$ convergence.

Next, one observes that the exact formula \eqref{eq:Uformula} can be written as the integral of the
function $2\arctan(\epsilon^{-1}\alpha)$ against the normalized counting measure of eigenvalues
of $\tilde{\mathbf{A}}$.  Pointwise, the integrand converges to $\pi\,\mathrm{sgn}(\alpha)$ as
$\epsilon\downarrow 0$, and by a careful dominated convergence argument one then establishes the desired
locally uniform convergence of $\tilde{U}$ to $V$.
\end{proof}

Before giving our next result, we recall the inviscid Burgers hierarchy.  Consider the equation
\begin{equation}
u_\mathrm{B} = u_0\left(x-\sum_{k=1}^K(k+1)u_\mathrm{B}^kt_k\right).
\label{eq:Burgersimplicit}
\end{equation}
For $t_1,\dots,t_K$ all sufficiently small (given $x\in\mathbb{R}$) it follows from the implicit function theorem that there exists
a unique solution $u_\mathrm{B}(x,t_1,t_2,\dots,t_K)\approx u_0(x)$.  As $t_k$ increases from zero,
there will be bifurcation points at which the number of solutions of \eqref{eq:Burgersimplicit}
increases by a finite even integer (this is due to the condition on inflection points of $f(x)=u_0(x)^k$
satisfied by admissible $u_0$).  Therefore, near a given $x$ and at given values of $t_1,t_2,\dots,t_K$, there will generically
be an odd finite number $2P(x,t_1,t_2,\dots,t_K)+1$ of distinct solutions $u_{\mathrm{B},0}<u_{\mathrm{B},1}<\dots<u_{\mathrm{B},2P}$ to \eqref{eq:Burgersimplicit}, and each is differentiable with
respect to all of the independent variables $x$ and $t_1,t_2,\dots,t_K$.  By differentiation of 
\eqref{eq:Burgersimplicit} one observes that each of the solution branches is a function $u_\mathrm{B}(x,t_1,t_2,\dots,t_K)$ that simultaneously satisfies the equations
\begin{equation}
\frac{\partial u_\mathrm{B}}{\partial t_k}+(k+1)u_\mathrm{B}^k\frac{\partial u_\mathrm{B}}{\partial x}=0\quad\text{or}\quad\frac{\partial u_\mathrm{B}}{\partial t} +\frac{\partial}{\partial x}u_\mathrm{B}^{k+1}=0,\quad k=1,2,\dots,K.
\label{eq:BurgersHierarchy}
\end{equation}
These are the partial differential equations of the inviscid Burgers hierarchy.  The simultaneous
solution of these equations with initial condition $u_\mathrm{B}(x,0,0,\dots,0)=u_0(x)$ is accomplished
by the method of characteristics and produces the solution in implicit form \eqref{eq:Burgersimplicit}.
We note that whereas typically when Burgers-type equations appear in the theory of partial differential equations one is interested in single-valued weak solutions representing shock waves, our interest here is in the multivalued solution produced by finding all real solutions of the implicit equation \eqref{eq:Burgersimplicit}.  See Figures~\ref{fig:IC} and \ref{fig:Burgers}.
\begin{figure}[h]
\begin{center}
\includegraphics[width=0.3\linewidth]{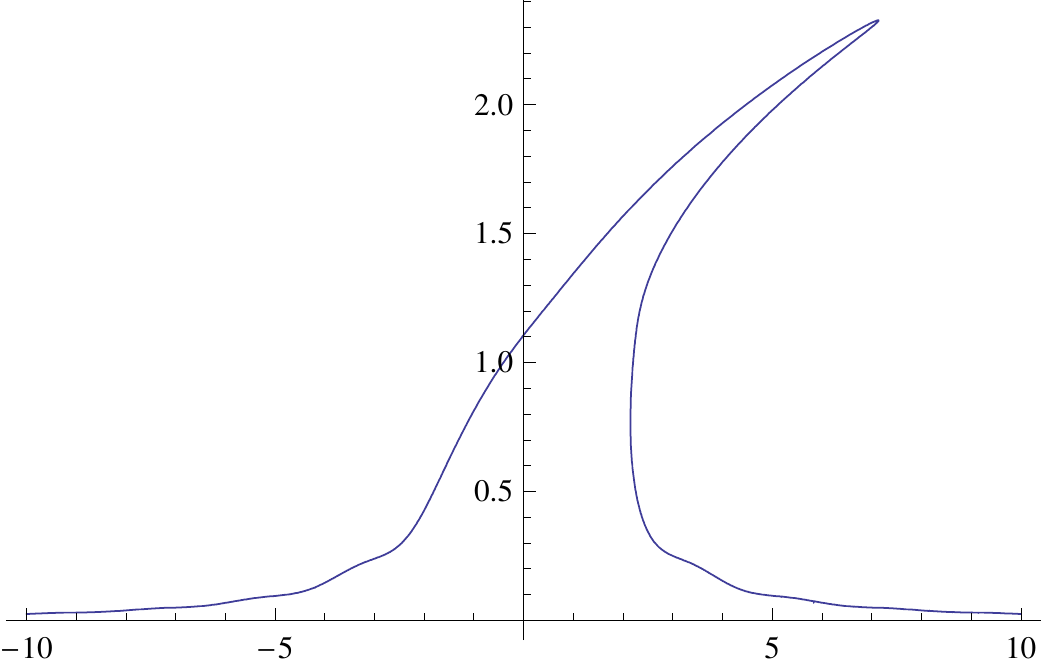}\hspace{0.049\linewidth}%
\includegraphics[width=0.3\linewidth]{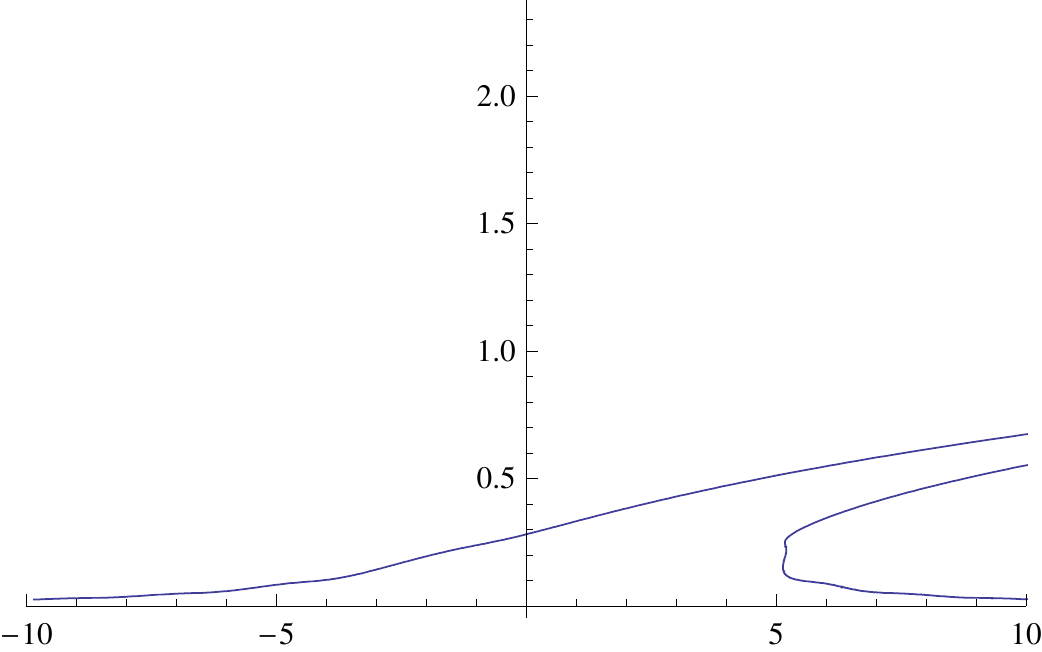}\hspace{0.049\linewidth}%
\includegraphics[width=0.3\linewidth]{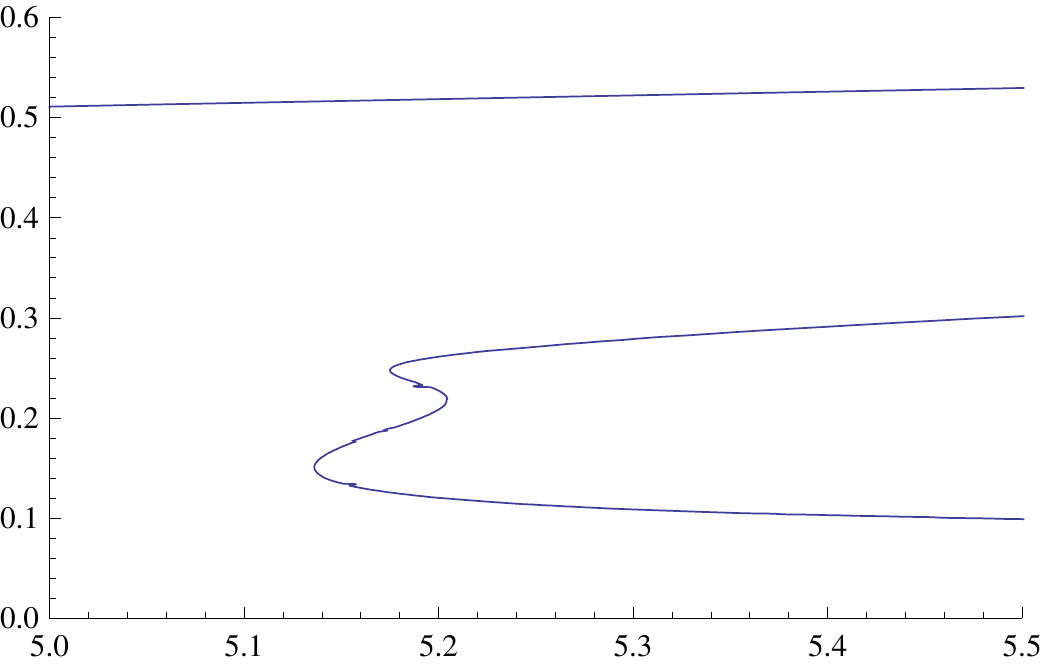}
\end{center}
\caption{Left and center:  the solution $u_\mathrm{B}$ of \eqref{eq:Burgersimplicit} with initial data as
specified in Figure~\ref{fig:IC} with $t_k=0$ for $k>3$ and with $t_1=t_2=t_3=0.1$ and $t_1=t_2=t_3=3$ respectively.  Right:  a close-up of the case when $t_1=t_2=t_3=3$ displaying  intervals of $x$ with
$P=0,1,2$ (one, three, and five branches, respectively).}
\label{fig:Burgers}
\end{figure}

\begin{proposition}  The function $V$ is of class $C^1(\mathbb{R}^{K+1})$, and 
\begin{equation}
\frac{\partial V}{\partial x}=\sum_{n=0}^{2P}(-1)^nu_{\mathrm{B},n}(x,t_1,t_2,\dots,t_K)
\end{equation}
while
\begin{equation}
-\frac{1}{m}\frac{\partial V}{\partial t_{m-1}} =\sum_{n=0}^{2P}
(-1)^n \frac{1}{m}u_{\mathrm{B},n}(x,t_1,t_2,\dots,t_K)^m,\quad m=2,3,\dots,K+1,
\end{equation}
where $P=P(x,t_1,t_2,\dots,t_K)$.
Both of these formulae assume that $(x,t_1,t_2,\dots,t_K)\in\mathbb{R}^{K+1}$ is a point at which
the integer $P$ is well-defined; however, since new solution branches bifurcate in pairs from the same
value of $u_\mathrm{B}$ it is clear that the formulae extend by continuity to all of $\mathbb{R}^{K+1}$.
\end{proposition}
\begin{proof}
The proof is virtually identical to that of Lemma 4.12 from \cite{MillerX10}.  The idea is as follows.
With the use of the explicit formula \label{eq:Gdef} for $G(\alpha;x,t_1,t_2,\dots,t_K)$ one can exchange the order of
integration in the formula \eqref{eq:Vdef} to obtain
\begin{equation}
V(x,t_1,t_2,\dots,t_K)=\int_{-L}^0 J(\lambda;x,t_1,t_2,\dots,t_K)\,d\lambda
\label{eq:VJ}
\end{equation}
where
\begin{equation}
J(\lambda;x,t_1,t_2,\dots,t_K):=\begin{cases}
-\pi F(\lambda),&\quad X(\lambda;x,t_1,t_2,\dots,t_K)< x_-(\lambda)\\
X(\lambda;x,t_1,t_2,\dots,t_K) +\gamma(\lambda),&\quad
x_-(\lambda)\le X(\lambda;x,t_1,t_2,\dots,t_K)\le x_+(\lambda)\\
\pi F(\lambda),&\quad X(\lambda;x,t_1,t_2,\dots,t_K)> x_+(\lambda).
\end{cases}
\end{equation}
The integrand therefore has a different form as a function of $\lambda$ in three different types
of subintervals of $[-L,0]$, with boundary points given by the solutions $\lambda$ of the equations
\begin{equation}
X(\lambda;x,t_1,t_2,\dots,t_K)=x_\pm(\lambda).
\end{equation}
Recalling that $x_\pm(\lambda)$ are two branches of the inverse function of $u_0$ in the sense that for $-L<\lambda<0$, $u_0(x_\pm(\lambda))=-\lambda$, both of these equations can be combined in the form
\begin{equation}
-\lambda=u_0(X(\lambda;x,t_1,t_2,\dots,t_K)),
\end{equation}
which one immediately notices is the same implicit equation \eqref{eq:Burgersimplicit} providing the multivalued solution of the Burgers hierarchy, under the substitution $u_\mathrm{B}=-\lambda$.
Differentiation of \eqref{eq:VJ} using Leibniz' rule to take into account the moving boundaries
then yields the desired formulae.
\end{proof}

Now we may formulate our main result.
\begin{theorem}
Let $v\in L^2(\mathbb{R})$.  Then
\begin{equation}
\lim_{\epsilon\downarrow 0} \int_\mathbb{R}\tilde{D}_1(x,t_1,t_2,\dots,t_K)v(x)\,dx = 
\int_\mathbb{R}\left(\sum_{n=0}^{2P}(-1)^nu_{\mathrm{B},n}(x,t_1,t_2,\dots,t_K)\right)v(x)\,dx
\label{eq:L2weak}
\end{equation}
holds uniformly for $(t_1,\dots,t_K)$ in compact subsets of $\mathbb{R}^K$.  Recalling that $\tilde{D}_1=\tilde{u}$, this means that $\tilde{u}$ converges to the alternating sum of branches of the multivalued solution of the Burgers hierarchy with initial condition $u_0$ in the weak $L^2(\mathbb{R}_x)$ sense.  Moreover, the convergence is in the strong $L^2(\mathbb{R}_x)$ sense if $t_1,\dots,t_K$ are all sufficiently small that $u_\mathrm{B}$ is
single-valued as a function of $x$.

Now let  $\phi\in\mathscr{D}(\mathbb{R})$ be a test function.  Then for $m=2,3,\dots,K+1$,
\begin{equation}
\lim_{\epsilon\downarrow 0}\int_\mathbb{R}\tilde{D}_m(x,t_1,t_2,\dots,t_K)\phi(t_{m-1})\,dt_{m-1}=
\int_\mathbb{R}\left(\sum_{n=0}^{2P}(-1)^n\frac{1}{m}u_{\mathrm{B},n}(x,t_1,t_2,\dots,t_K)^m\right)\phi(t_{m-1})\,dt_{m-1}
\end{equation}
holds uniformly for $(x,t_1,\dots,t_{m-2},t_{m},\dots,t_K)$ in compact subsets of $\mathbb{R}^K$.
Therefore $\tilde{D}_m$ converges to the alternating sum of $m$-th powers, weighted by $1/m$, of the branches of the multivalued solution of the Burgers hierarchy with initial condition $u_0$ in the topology 
of $\mathscr{D}'(\mathbb{R}_{t_{m-1}})$, that is, distributional convergence with respect to $t_{m-1}$.
\end{theorem}
\begin{proof}
The distributional convergence of $\tilde{D}_m$ for $m\ge 2$ clearly follows from our above results.
It is also easy to conclude that \eqref{eq:L2weak} holds if $v$ is specialized to a test function $\phi\in\mathscr{D}(\mathbb{R})$.  To strengthen this to weak $L^2(\mathbb{R}_x)$ convergence and strong
$L^2(\mathbb{R}_x)$ convergence pre-breaking, one follows nearly verbatim the arguments on pages 254--256 of
\cite{MillerX10}.
\end{proof}

We expect that with some additional effort, the nature of the convergence of $\tilde{D}_m$ for $m\ge 2$
can be strengthened to exactly the same type as is available for $\tilde{D}_1=\tilde{u}$, a type of convergence that is more suitable for evaluation at a point in the phase space $\mathscr{P}$ of fields.
This expectation is based on the reasonable hypothesis that the weak (or distributional) nature of the convergence stems from the presence of wild oscillations that can be modeled by modulated $P$-phase wave exact solutions of the Benjamin-Ono hierarchy as have been described by Matsuno \cite{Matsuno} using
the bilinear method of Hirota.  These $P$-phase waves have also been obtained directly from the Lax pair for \eqref{eq:BO} (the $k=1$ case of the hierarchy only) by Dobrokhotov and Krichever \cite{DobrokhotovK91}, who further provided a formal Whitham-type modulation theory for these waves,
noting that the modulation equations simply take the form of $2P+1$ copies of the inviscid Burgers
equation (equation \eqref{eq:BurgersHierarchy} for $k=1$).  In light of our results, it appears that these $2P+1$ copies should be globally viewed as sheets of the same multivalued solution.  In any case,
if one interprets the distributional limits in $\mathscr{D}'(\mathbb{R}_{t_{m-1}})$ as local averages of $\tilde{D}_m$ over vanishingly small intervals of $t_{m-1}$, then assuming only that the wavenumbers and frequencies are not rationally dependent, these averages could just as well be calculated over
small intervals in $x$, holding $t_1,t_2,\dots,t_K$ fixed.  In other words, if the weak limits are necessary
due to the presence of modulated multiphase waves of wavelengths and periods proportional to
$\epsilon$, then there should at generic points be no difference between convergence in $\mathscr{D}'(\mathbb{R}_{t_{m-1}})$ and convergence in $\mathscr{D}'(\mathbb{R}_x)$.  A proof 
of such a result probably requires resolution of the microstructure as could be obtainable from
an approach to the zero-dispersion limit that starts with the nonlocal Riemann-Hilbert problem of inverse scattering for Benjamin-Ono,  and that involves the development of  some new analogue of the Deift-Zhou asymptotic method as has been applied \cite{DeiftVZ97} to strengthen the zero-dispersion limit of Korteweg-de Vries equation.  We hope to be able to announce progress in this direction in the near future.

It seems to us that while the Benjamin-Ono equation looks at first glance to be a more complicated model for wave propagation than the more famous Korteweg-de Vries equation due to the presence of the Hilbert transform and its concomitant nonlocality (and perhaps even at ``second glance'', since the
treatment of the Benjamin-Ono equation by the inverse-scattering transform method is far less well-understood than in the case of the Korteweg-de Vries equation), in fact it is far simpler in the zero-dispersion limit.  Indeed, the asymptotic formulae that are the analogues in the Korteweg-de Vries case
of our limiting formulae for $\tilde{D}_m$ require the solution  of a variational problem for a quadratic functional with constraints as was found by Dave Levermore and Peter Lax in their pioneering work \cite{LaxL83}, while for Benjamin-Ono it suffices to be able to solve the implicit algebraic equation \eqref{eq:Burgersimplicit} for $u_\mathrm{B}$, or alternatively to solve the system of partial differential equations \eqref{eq:BurgersHierarchy} numerically by the method of characteristics.  These
are far more elementary tasks.  We want to stress this point to hopefully encourage the use in the practical modeling of internal waves of the simple approximate formulae available for the Benjamin-Ono equation and its hierarchy when the dispersion parameter $\epsilon$ can be reasonably assumed to be small.

\end{document}